\documentclass{article}
\PassOptionsToPackage{numbers, compress}{natbib}

\usepackage[preprint]{neurips_2020}

\usepackage[utf8]{inputenc} 
\usepackage[T1]{fontenc}    
\usepackage{hyperref}       
\usepackage{url}            
\usepackage{booktabs}       
\usepackage{amsfonts}       
\usepackage{nicefrac}       
\usepackage{microtype}      
\usepackage{graphicx, caption, subcaption}

\usepackage{bm}
 \usepackage{algorithm,algorithmic}

\usepackage{mathtools}
\usepackage{amsmath}
\usepackage[shortlabels]{enumitem}

\usepackage{graphicx,epic,eepic,epsfig,amsmath,latexsym,amssymb,verbatim,color}
 
\usepackage{amsfonts}       
\usepackage{nicefrac}       

\usepackage{amsmath}
\usepackage{bbm}

\usepackage{float}
\usepackage{tikz}
\usetikzlibrary{chains}
\usetikzlibrary{fit}
\usepackage{pgflibraryarrows}		
\usepackage{pgflibrarysnakes}		

\usepackage{epsfig}
\usetikzlibrary{shapes.symbols,patterns} 
\usepackage{pgfplots}

\usepackage[strict]{changepage}
\usepackage{hyperref}
\hypersetup{colorlinks=true,citecolor=blue,linkcolor=blue,filecolor=blue,urlcolor=blue,breaklinks=true}

\usepackage[marginal]{footmisc}
\usepackage{url}
\usepackage{theorem}
\newtheorem{definition}{Definition}
\newtheorem{proposition}[definition]{Proposition}
\newtheorem{lemma}[definition]{Lemma}

\newtheorem{theorem}[definition]{Theorem}
\newtheorem{corollary}[definition]{Corollary}

\def\squareforqed{\hbox{\rlap{$\sqcap$}$\sqcup$}}
\def\qed{\ifmmode\squareforqed\else{\unskip\nobreak\hfil
\penalty50\hskip1em\null\nobreak\hfil\squareforqed
\parfillskip=0pt\finalhyphendemerits=0\endgraf}\fi}
\def\endenv{\ifmmode\;\else{\unskip\nobreak\hfil
\penalty50\hskip1em\null\nobreak\hfil\;
\parfillskip=0pt\finalhyphendemerits=0\endgraf}\fi}
\newenvironment{proof}{\noindent \textbf{{Proof~} }}{\hfill $\blacksquare$}

\newcounter{remark}
\newenvironment{remark}[1][]{\refstepcounter{remark}\par\medskip\noindent%
\textbf{Remark~\theremark #1} }{\medskip}

\newcounter{example}

\mathchardef\ordinarycolon\mathcode`\:
\mathcode`\:=\string"8000
\def\vcentcolon{\mathrel{\mathop\ordinarycolon}}
\begingroup \catcode`\:=\active
  \lowercase{\endgroup
  \let :\vcentcolon
  }

\usepackage{cleveref}
\usepackage{graphicx}
\usepackage{xcolor}

\RequirePackage[framemethod=default]{mdframed}
\newmdenv[skipabove=7pt,
skipbelow=7pt,
backgroundcolor=darkblue!15,
innerleftmargin=5pt,
innerrightmargin=5pt,
innertopmargin=5pt,
leftmargin=0cm,
rightmargin=0cm,
innerbottommargin=5pt,
linewidth=1pt]{tBox}

\newmdenv[skipabove=7pt,
skipbelow=7pt,
backgroundcolor=blue2!25,
innerleftmargin=5pt,
innerrightmargin=5pt,
innertopmargin=5pt,
leftmargin=0cm,
rightmargin=0cm,
innerbottommargin=5pt,
linewidth=1pt]{dBox}
\newmdenv[skipabove=7pt,
skipbelow=7pt,
backgroundcolor=darkkblue!15,
innerleftmargin=5pt,
innerrightmargin=5pt,
innertopmargin=5pt,
leftmargin=0cm,
rightmargin=0cm,
innerbottommargin=5pt,
linewidth=1pt]{sBox}
\definecolor{darkblue}{RGB}{0,80,156}
\definecolor{darkkblue}{RGB}{0,0,153}
\definecolor{blue2}{RGB}{102,178,255}
\definecolor{darkred}{RGB}{195,0,0}

\newcommand{\nc}{\newcommand}
\nc{\rnc}{\renewcommand}
\nc{\beg}{\begin{equation}}
\nc{\eeq}{{\end{equation}}}
\nc{\beqa}{\begin{eqnarray}}
\nc{\eeqa}{\end{eqnarray}}
\nc{\lbar}[1]{\overline{#1}}
\nc{\bra}[1]{\langle#1|}
\nc{\ket}[1]{|#1\rangle}
\nc{\ketbra}[2]{|#1\rangle\!\langle#2|}
\nc{\braket}[2]{\langle#1|#2\rangle}

\nc{\proj}[1]{| #1\rangle\!\langle #1 |}
\nc{\avg}[1]{\langle#1\rangle}
\nc{\rank}{\operatorname{Rank}}
\nc{\smfrac}[2]{\mbox{$\frac{#1}{#2}$}}
\nc{\tr}{\operatorname{Tr}}
\nc{\ox}{\otimes}
\nc{\dg}{\dagger}
\nc{\dn}{\downarrow}
\nc{\cA}{{\cal A}}
\nc{\cB}{{\cal B}}
\nc{\cC}{{\cal C}}
\nc{\cD}{{\cal D}}
\nc{\cE}{{\cal E}}
\nc{\cF}{{\cal F}}
\nc{\cG}{{\cal G}}
\nc{\cH}{{\cal H}}
\nc{\cI}{{\cal I}}
\nc{\cJ}{{\cal J}}
\nc{\cK}{{\cal K}}
\nc{\cL}{{\cal L}}
\nc{\cM}{{\cal M}}
\nc{\cN}{{\cal N}}
\nc{\cO}{{\cal O}}
\nc{\cP}{{\cal P}}
\nc{\cQ}{{\cal Q}}
\nc{\cR}{{\cal R}}
\nc{\cS}{{\cal S}}
\nc{\cT}{{\cal T}}
\nc{\cV}{{\cal V}}
\nc{\cX}{{\cal X}}
\nc{\cY}{{\cal Y}}
\nc{\cZ}{{\cal Z}}
\nc{\cW}{{\cal W}}
\nc{\csupp}{{\operatorname{csupp}}}
\nc{\qsupp}{{\operatorname{qsupp}}}
\nc{\var}{{\operatorname{var}}}
\nc{\rar}{\rightarrow}
\nc{\lrar}{\longrightarrow}
\nc{\polylog}{{\operatorname{polylog}}}
\nc{\wt}{{\operatorname{wt}}}
\nc{\av}[1]{{\left\langle {#1} \right\rangle}}
\nc{\supp}{{\operatorname{supp}}}

\nc{\argmin}{{\operatorname{argmin}}}

\def\x{\xi}

\nc{\RR}{{{\mathbb R}}}
\nc{\CC}{{{\mathbb C}}}
\nc{\FF}{{{\mathbb F}}}
\nc{\NN}{{{\mathbb N}}}
\nc{\ZZ}{{{\mathbb Z}}}
\nc{\PP}{{{\mathbb P}}}
\nc{\QQ}{{{\mathbb Q}}}
\nc{\UU}{{{\mathbb U}}}
\nc{\EE}{{{\mathbb E}}}
\nc{\id}{{\operatorname{id}}}

\nc{\CHSH}{{\operatorname{CHSH}}}

\nc{\be}{\begin{equation}}
\nc{\ee}{{\end{equation}}}
\nc{\bea}{\begin{eqnarray}}
\nc{\eea}{\end{eqnarray}}
\nc{\<}{\langle}
\rnc{\>}{\rangle}
\nc{\rU}{\mbox{U}}

\nc{\ob}[1]{#1}

\nc{\SEP}{{\text{\rm SEP}}}
\nc{\NS}{{\text{\rm NS}}}
\nc{\LOCC}{{\text{\rm LOCC}}}
\nc{\PPT}{{\text{\rm PPT}}}
\nc{\EXT}{{\text{\rm EXT}}}
\nc{\Sym}{{\operatorname{Sym}}}

\nc{\ERLO}{{E_{\text{r,LO}}}}
\nc{\ERLOCC}{{E_{\text{r,LOCC}}}}
\nc{\ERPPT}{{E_{\text{r,PPT}}}}
\nc{\ERLOCCinfty}{{E^{\infty}_{\text{r,LOCC}}}}
\nc{\Aram}{{\operatorname{\sf A}}}

\usepackage{tikz}
\usepackage{hyperref}
\hypersetup{colorlinks=true,citecolor=blue,linkcolor=blue,filecolor=blue,urlcolor=blue,breaklinks=true}

\makeatletter
\def\grd@save@target#1{%
  \def\grd@target{#1}}
\def\grd@save@start#1{%
  \def\grd@start{#1}}
\tikzset{
  grid with coordinates/.style={
    to path={%
      \pgfextra{%
        \edef\grd@@target{(\tikztotarget)}%
        \tikz@scan@one@point\grd@save@target\grd@@target\relax
        \edef\grd@@start{(\tikztostart)}%
        \tikz@scan@one@point\grd@save@start\grd@@start\relax
        \draw[minor help lines,magenta] (\tikztostart) grid (\tikztotarget);
        \draw[major help lines] (\tikztostart) grid (\tikztotarget);
        \grd@start
        \pgfmathsetmacro{\grd@xa}{\the\pgf@x/1cm}
        \pgfmathsetmacro{\grd@ya}{\the\pgf@y/1cm}
        \grd@target
        \pgfmathsetmacro{\grd@xb}{\the\pgf@x/1cm}
        \pgfmathsetmacro{\grd@yb}{\the\pgf@y/1cm}
        \pgfmathsetmacro{\grd@xc}{\grd@xa + \pgfkeysvalueof{/tikz/grid with coordinates/major step}}
        \pgfmathsetmacro{\grd@yc}{\grd@ya + \pgfkeysvalueof{/tikz/grid with coordinates/major step}}
        \foreach \x in {\grd@xa,\grd@xc,...,\grd@xb}
        \node[anchor=north] at (\x,\grd@ya) {\pgfmathprintnumber{\x}};
        \foreach \y in {\grd@ya,\grd@yc,...,\grd@yb}
        \node[anchor=east] at (\grd@xa,\y) {\pgfmathprintnumber{\y}};
      }
    }
  },
  minor help lines/.style={
    help lines,
    step=\pgfkeysvalueof{/tikz/grid with coordinates/minor step}
  },
  major help lines/.style={
    help lines,
    line width=\pgfkeysvalueof{/tikz/grid with coordinates/major line width},
    step=\pgfkeysvalueof{/tikz/grid with coordinates/major step}
  },
  grid with coordinates/.cd,
  minor step/.initial=.2,
  major step/.initial=1,
  major line width/.initial=2pt,
}
\makeatother

\usepackage{thmtools}
\usepackage{thm-restate}
\usepackage{etoolbox}
\makeatletter
\def\problem@s{}
\newcounter{problems@cnt}

\newcommand{\allproblems}{\problem@s}
\makeatother

\usepackage[qm]{qcircuit}
\usepackage{array}

\newcommand{\trace}{\operatorname{Tr}}

\DeclareMathOperator*{\argmax}{arg\,max}

\newcommand{\order}{\operatorname{O}}
\newcommand{\Reg}{\operatorname{Reg}}
\title{More Practical and Adaptive Algorithms for Online Quantum State Learning}

\author{%
 Yifang Chen \\
 Institute for Quantum Computing\\
 Baidu Research, Beijing 100193, China \\
 \texttt{yifang@usc.edu} \\
 \And
 Xin Wang \\
 Institute for Quantum Computing\\
 Baidu Research, Beijing 100193, China \\
 \texttt{wangxin73@baidu.com} \\
}

\begin{document}

\maketitle

\begin{abstract}
Online quantum state learning is a recently proposed problem by Aaronson et al. (2018), where the learner sequentially predicts $n$-qubit quantum states based on given measurements on states and noisy outcomes. 
In the previous work, the algorithms are worst-case optimal in general, but fail in achieving tighter bounds in certain simpler or more practical cases. In this paper, we develop algorithms to advance the online learning of quantum states.
First, we show that Regularized Follow-the-Leader (RFTL) method with Tallis-2 entropy can achieve an $\order(\sqrt{MT})$ total loss with perfect hindsight on the first $T$ measurements with maximum rank $M$. This regret bound depends only on the maximum rank $M$ of measurements rather than the number of qubits, which takes advantage of low-rank measurements. 
Second, we propose a parameter-free algorithm based on a classical adjusting learning rate schedule that can achieve a regret depending on the loss of best states in hindsight, which takes advantage of low noisy outcomes. 
Besides these more adaptive bounds, we also show that our RFTL with Tallis-2 entropy algorithm can be implemented efficiently on near-term quantum computing devices, which is not achievable in previous works. 
\end{abstract}

\section{Introduction}
Major academic and industry efforts are currently in progress to realize scalable quantum hardware and develop powerful quantum software. Quantum computers are expected to have significant applications in solving intractable problems in areas including optimization, chemistry, security, and machine learning. On the other hand, machine learning theories may also help us better solve the problems in quantum computation~\cite{Biamonte2017c,Dunjko2018,Ciliberto2018}.  
One of the most fundamental topics in quantum computation and quantum information is to learn an unknown quantum state. Sufficient knowledge of a quantum state is in particular vital and indispensable in verification and identification of experimental outcomes on near-term quantum devices.

Suppose we have a quantum device or a physical process that could produce a quantum state. By using the device repeatedly, we can prepare many copies of the state and can then measure each copy. The goal of quantum state learning is to learn an approximate description of the state based on various measurement outcomes. To obtain a full characterization of an unknown quantum state, the most well-known approach is to do quantum state tomography~\cite{Nielsen2010}, which is of great practical and theoretical importance. To be specific, the goal of state tomography is to reconstruct the full density matrix that approximates the target unknown state $\rho$ within $\varepsilon$ in trace distance. In general, a $n$-qubit quantum state is described by roughly $2^{2n}$ real parameters, and a full tomography of these parameters is quite costly. Fully reconstruction of an unknown state in the worst-case costs exponential copies of the state~\cite{ODonnell2016,Haah2017b}.
The cost for tomography of an arbitrary $50$-qubit state is already prohibitively expensive.
 
However, some side information with respect to measurements is sufficient in some applications.
Therefore, \cite{Aaronson2007} considered how well a quantum state $\rho$ can be learned from the results of measurements. The goal of this setting is to construct a quantum state $\sigma$ that has roughly the same expectation value as the target state $\rho$ for a collection of two-outcome measurements. Recently, \cite{Aaronson2018b} considered the problem of shadow tomography, where the learner aims to output  $\tr(E_j\rho)$ up to an additive error. Under this setting, \cite{Aaronson2018b} showed that the required number of copies of the unknown state is nearly linear to the number of qubits and poly-logarithmic in terms of the number of the measurements.

One drawback of the above learning methods is the assumption that the sample measurements are drawn independently from some probability distribution, which is unrealistic in practice and fails to work for adversarial environments. Such an assumption does not always hold since the environment in experiments may change over time. It is thus desirable to design quantum learning algorithms that work in a more general model with sequential or even adversarial data.

The online learning theory~\cite{Shalev-Shwartz2011a,Hazan2016} represents a powerful tool to handle adversarial sequential data, which has a profound impact on machine learning. 
Inspired by the success of online learning, developing their quantum counterparts is a natural and essential topic in the emerging field of quantum machine learning.
Recently, \cite{Aaronson2018a} proposed the model of online quantum states learning to overcome the above drawback of the i.i.d. assumption and deal with sequential data. In the online learning model, the data of measurements and losses are provided sequentially, and the referee is able to choose the measurement operator even adversarially. The leaner predicts a series of quantum states interactively in the learning process. The goal is to minimize the regret, which is the difference in the total loss between the learning algorithm and the best hypothesis quantum state in hindsight. 
 
In this paper, we focus on the problem of online quantum state learning and consider more practical and adaptive cases. First, existing online learning algorithms are entirely classical, so we here consider a quantum algorithm that can be efficiently implemented on near-term quantum devices. Besides, most existing algorithms are proven to be worst-case optimal, but they would not perform well when dealing with a practical problem that is probably not the worst case or even relatively easy. Here we consider several common quantum computing assumptions:
\begin{itemize}
    \item First, the existing algorithms assume that the rank of the measurement operator is arbitrary and unknown, while in practice, the rank might be low and known in advance. 
    
    \item Second, the existing algorithms fail in adapting to the case when the loss of the best quantum state in hindsight is small. In quantum, this happens when each round we are allowed to measure a large batch of states copies each gives a binary outcome, and therefore the average we observed will close to the true "Yes" probability. This includes the special  "realizable"  case where the feedback is always the real "Yes" probability, which means the loss of best quantum states in hindsight is always zero.
\end{itemize}

In this work, we propose several algorithms to address those problems.
\begin{itemize}
    \item We propose the Regulerized Follow-the-Leader (RFTL) algorithm with Tallis-2 Entropy that achieves the regret in term of measurement rank $M$ as $\operatorname{O} (\sqrt{MT})$. Moreover, we propose a variational quantum algorithm that can achieve the same result on the near-term quantum device. (Section~\ref{sec: tsallis2})
    \item When the loss is $L_2$-norm, we provide a parameter-free learning rate adjusting strategy, which helps both RFTL with von Neumann entropy and Matrix Exponentiated Gradient (MEG) updating methods to achieve a so-called "small loss" bound $\Tilde{\operatorname{O}}\left(\sqrt{n L^*} + n \right)$, where $L^*$ is the loss of the best quantum state in hindsight and is formally defined in Eq.~\eqref{eq: def of small loss}. Notice that \cite{Aaronson2018a} provides a worst-case optimal algorithm based on RFTL with von-Neumann entropy, but whether we can obtain a "small loss" bound using the same algorithm remains open. In addition, \cite{Tsuda2005} shows that MEG can achieve such a "small loss" bound only when $L^*$ is known in advance, and the adversary is oblivious. Therefore, our result is strictly better than these two. (Section~\ref{sec: small_loss})
\end{itemize}

\section{Preliminaries and Problem Description}
\label{sec:setup}
\paragraph{Quantum states}
In quantum computing, a quantum bit (or qubit) could be in a superposition of $0$ and $1$, which is different with the classical bit in a deterministic state of either $0$ or $1$. Generally, any quantum data can be represented by quantum states over some space $\cX$ (e.g., $\cX = \CC^d$).
The pure states and mixed states can be unitedly described by a mathematical tool called density matrix (or density operator), which is a positive semidefinite matrix (i.e., $\rho \ge 0$) with trace one (i.e., $\tr\rho =1$). The set of $d$-dimensional quantum states is denoted by $\cS_d$. A quantum state is pure if $\rank(\rho)=1$; otherwise it is a mixed state. Particularly, a pure state can be represented by a unit vector in the sense that $\rho = \proj{\psi}$ with $\ket \psi \in \CC^d$ and  $\bra\psi = \ket \psi^{\dagger}$, where $\dagger$ refers to conjugate transpose.
Quantum states on a composed system of $\cX$ and $\cY$ can be described by density operators over the tensor (or Kronecker) product space $\cX\otimes \cY$. 
For a bipartite state $\rho_{XY}$, its marginal state $\rho_X$ on system $\cX$  can characterized via the partial trace function $\tr_{Y}(\cdot)$, i.e., $\rho_X = \tr_{Y}(\rho_{XY})$.

\paragraph{Two-outcome measurements}
Two-outcome measurement Quantum measurement is
a means to extract classical (observable) information from
a quantum state. When applying a two-outcome measurement to a quantum state, it succeeds with a specific probability. Mathematically, a two-outcome measurement could be described
by a Hermitian matrix $E$ with eigenvalues in $[0,1]$. When applying $E$ to a quantum state $\rho$, it outputs a successful result ``Yes'' with probability $\tr E\rho$ and a rejected
result ``No'' with probability $1-\tr E\rho$.

\paragraph{Frobenius norm and its dual}
Given two complex vector spaces $\mathcal{X}=\mathbb{C}^{d}$ and $\mathcal{Y}=\mathbb{C}^{m}$,
we denote the space of all linear operators of the form $A:\cX\to\cY$ as $L(\cX, \cY)$. The notation $L(\cX)$ is shorthand for $L(\cX, \cX)$. An inner product on $L(\cX, \cY)$ is defined as $\langle A, B\rangle=\operatorname{Tr}\left(A^{\dagger} B\right)$, where $A^\dagger$ is the conjugate transpose of $A$. 

For $\mathcal{X}=\mathbb{C}^{d},\mathcal{Y}=\mathbb{C}^{m}$, and any operator $A\in L(\cX,\cY)$, the Frobenius norm is defined as 
\begin{align}
\|A\|_2 = \sqrt{\tr(A^{\dagger}A)}.
\end{align}
This norm corresponds precisely to the $2$-norm of the vector of singular values of $A$. In particular, for a Hermitian operator $A$ ($A = A^{\dagger}$), its the Frobenius norm simplifies to 
$\|A\|_2 = \sqrt {\tr A^2}$. Moreover, the dual norm of the Frobenius norm is itself, i.e.,
$\|A\|_2^* = \|A\|_2$.

\paragraph{Tsallis entropy} For a quantum state $\rho$, the quantum Tsallis entropy is defined as follows
\begin{align}
S_{q}(\rho):=\frac{1}{1-q}\left(\operatorname{Tr}\left(\rho^{q}\right)-1\right), q \in(0,1) \cup(1, \infty).
\end{align}
In particular, the Tsallis-$2$ entropy is closely connected to the Frobenius norm in the sense that 
\begin{align}
S_{2}(\rho) = 1 - \tr \rho^2 = 1 - \|\rho\|_2^2.
\end{align}



\paragraph{Online learning and problem set-up} 
Online quantum state learning is a sequential prediction process with interactions between a player and a referee or adversary over $T$ rounds. Suppose there is an underlying unknown state $\rho$ to be learned. At the round $t$, the player predicts
a quantum state $\omega_t$. At the same time, the referee or adversary adaptively
chooses and reveals a two-outcome measurement $E_t$ and an outcome feedback $b_t$, based on the learner's strategy and the past history. Here $b_t$ could be the exact value or an approximate value of $\ell_t(\tr E_t \rho)$ with random noise, but is allowed to be arbitrary in general. At the end of this round, the player
suffers a loss 
\begin{align}
\ell_t(\tr E_t\omega_t),
\end{align}
where $\ell_t$ is a loss function. As a convention in online learning, we assume the loss function $\ell_{t}$ is convex and $L$-Lipschitz. Important examples of loss functions are $L_1$ loss defined as $\ell_{t}(z):=\left|z-b_{t}\right|$, and $L_2$ loss defined as $\ell_{t}(z):=(z-b_{t})^2$. 
 
The regret is the difference in the total loss between the
learning algorithm and the best hypothesis state in hindsight. Notice that this best hypothesis state in hindsight might be close but different from $\rho$ when feedback $b_t$ is noisy.
Formally this is captured by the notion of regret $Reg_T$, defined as 
\begin{align}
\Reg_T := \sum_{t=1}^T \ell_t(\tr E_t\omega_t)- \min_{\omega\in\cS_n}\sum_{t=1}^T \ell_t(\tr E_t\omega).
\end{align} 
 In this paper, we also consider the setting in which there are constraints on the measurement operators. For example, in the limited-rank setting,  the sequence of two-outcome measurements $E_t$ can be arbitrary but with limited rank $M$, i.e., $\rank(E_t)\le M$, $\forall t$. The limited-rank measurements are extensively considered and studied in the literature~\cite{Bruss2001,Jozsa2003,Eldar2002}.
 
In addition, we also consider the regret in terms of the loss of best hypothesis state in hindsight $L^*_T$ defined as 
\begin{align}
    L^*_T = \min_{\omega\in\cS_n}\sum_{t=1}^T \ell_t(\tr E_t\omega), \label{eq: def of small loss}
\end{align}
instead of the regret in terms of $T$.
In online learning, we usually call such regret form "small loss" bound. Some previous works about this small loss include~\cite{de_rooij_follow_2013} and ~\cite{Tsuda2005}.

\section{Learning with Tsallis-2 regularizer}
\label{sec: tsallis2}
\subsection{RFTL algorithm with Tsallis-2 regularizer and results}
We first follow the template of the Regularized Follow-the-Leader (RFTL) algorithm (see, e.g.,\cite{Hazan2016,Shalev-Shwartz2007,Abernethy2009}). In particular, we choose the Tsallis-2 entropy as the regularization term (or regularizer), 
while previous approach~\cite{Aaronson2018a} utilizes the von Neumann entropy.

\begin{algorithm}[H] 
\caption{RFTL algorithm with Tsallis-2 regularizer}
\begin{algorithmic}[1] \label{alg:rftl}
\STATE Input: number of rounds $T$, $\eta<\frac{1}{2}$, dimension $d$.
\STATE Set $\omega_1 = I/d$
\FOR{$t = 1, \ldots, T$}
\STATE Predict $\omega_{t}$. Consider the loss function $\ell_t$ and compute $\nabla_t = \ell_t'(\tr E_t\omega_t)E_t$.
\STATE Update $\omega_{t+1}$ via solving the following optimization
\begin{align}\label{eq: Tsallis-2 update rule}
\omega_{t+1}=\underset{\omega \in \cS_d}{\arg \min }\left\{\eta \sum_{s=1}^{t} \tr\left(\nabla_{s} \omega\right)-S_2(\omega)\right\}
\end{align}
\ENDFOR
\end{algorithmic}
\end{algorithm}

\begin{proposition}
\label{prop: sdp}
The updated state $\omega_t$ in Eq.~\eqref{eq: Tsallis-2 update rule} can be solved via seimidefinite programming.
\end{proposition}
Semidefinite programs~\cite{Vandenberghe1996} can be solved efficiently by polynomial-time algorithms~\cite{Khachiyan1980,AHK05,AK07,AHK12} and have many applications in combinatorial optimization  and quantum information (e.g., \cite{Jain2011a,Watrous2009,Wang2018}). Therefore the updated state $\omega_t$ in Eq.~\eqref{eq: Tsallis-2 update rule} can be computed efficiently in polynomial time in the dimension of the state. The CVX software~\cite{Grant2008} allows us to compute SDPs in practice. 

We remark that Algorithm 1 is more practical on near-term quantum computers than previous online learning methods~\cite{Aaronson2018a,Yang2020} in the sense that the predicted state $\omega_t$ in Eq.~\eqref{eq: Tsallis-2 update rule} can be prepared using parameterized quantum circuits. We discuss the details in Section~\ref{sec:near term}.

Furthermore, we show that Algorithm \ref{alg:rftl} achieves an $\order(\sqrt{T})$ regret bound, which is stated formally in Theorem \ref{thm:regret bound}. A detailed proof is provided in Appendix \ref{app: tallis2}.
\begin{theorem}\label{thm:regret bound}
Suppose that the maximum Frobenius norm of $\{E_t\}_{t=1}^T$ is $\lambda$, i.e., $\lambda= \max_t \|E_t\|_2$. By setting $\eta=\frac{1}{L\lambda\sqrt T}$ the regret of Algorithm $1$ is bounded by 
\begin{align}\label{:eq:formal bound}
\Reg_T\le 2L\lambda\sqrt{T}.
\end{align}
where $T$ is the number of rounds and $L$ is the Lipschitz constant of the loss function. 
\end{theorem}

This theoretical bound on the regret provides powerful and flexible estimations on the performance of Algorithm 1. On one hand, it achieves an $O(\sqrt{T})$ regret bound, and thus the average regret per round is approaching to $0$ when $T$ goes to infinity. On the other hand, it could assess the performance of regret for measurement operators with specific constraints or conditions. This is well-motivated due to two reasons. First, practically, the ability to perform measurements is limited, and there are usually technical limitations with the near-term quantum hardware. Second, theoretically, there are particular classes of measurements of great interest, such as projective or von Neumann measurements. In the following Corollary, we apply our estimation of regret to the case of limited-rank measurements and then generalize the estimation for general measurements.
\begin{corollary}
\label{coro: tallis2}
Suppose that the rank of each measurement operators is bounded by $M$. Then by setting $\eta = \frac{1}{L\sqrt{TM}}$, the regret of Algorithm 1 satisfies that
\begin{align}\label{eq:rank bound}
\Reg_T\le 2L\sqrt{MT}.
\end{align} 
For general measurement operators, we have
 \begin{align}\label{eq:dimension bound}
\Reg_T \le 2L\sqrt{dT}.
\end{align} 
\end{corollary}

\begin{remark}
\label{remark: relation with OGD}
We remark that the Online Gradient Descent (OGD) algorithm proposed in \cite{Yang2020} could also achieve the same regret bound as the RFTL method with Tsallis-2 regularizer by more careful analysis. But their algorithm is less efficient than ours, and they haven't discussed its implementation on quantum computing devices.
\end{remark}

    

\subsection{Variational quantum algorithm for state prediction}\label{sec:near term}
We further introduce a variational quantum algorithm to generate the state in Eq.~\eqref{eq: Tsallis-2 update rule}. In particular, the Tsallis-2 regularizer can be estimated efficiently via the well-known Swap Test~\cite{Buhrman2001,Gottesman2001a} or the recently developed Destructive Swap Test~\cite{Subasi2019}, which allows us to compute the state overlap between two quantum states $\rho$ and $\sigma$. Compared to the general Swap test, the Destructive Swap Test is more practical on near-term devices such as IBM’s and Rigetti’s quantum computers~\cite{Cross2017b,Smith2016}, since it is ancilla-free and costs less circuit depth and the number of the gates. More details could be found in~\cite{Subasi2019}.
Therefore, we introduce the following variational quantum algorithm for state prediction, which could be implemented on near-term quantum devices.

\begin{algorithm}[H] 
\caption{Variational quantum algorithm for state prediction}
\begin{algorithmic}[1]  
\STATE Input: choose the ansatz of unitary $U(\bm\theta)$, tolerance $\varepsilon$, and initial parameters of $\bm\theta$;

\STATE Apply  $U(\bm\theta)$ to three copies of the initial state $\ket{00}_{AR}$, and obtain three copies of the marginal states $\rho_{A_1} = \rho_{A_2} = \rho_{A_3} = \tr_R U(\bm\theta)\proj{00}_{AR}U(\bm\theta)^\dagger$; 
 
\STATE Compute  $\tr (H\rho_{A_1})$  and update cost function $L_1 =  \tr (H\rho_{A_1} )$, with $H = \eta \sum_{s=1}^{t-1} \tr\left(\nabla_{s} \rho_{A_1}\right)$;

\STATE Compute the state overlap $\tr (\rho_{A_2}\rho_{A_3})$ and update the cost function $L_2 =\tr(\rho_{A_2}\rho_{A_{3}})-1$;  

\STATE Perform optimization of $\cL(\bm\theta) = L_1+L_2$ and update parameters of $\bm\theta$;
\STATE Repeat 2-5 until the cost function $\cF_2(\bm\theta)$ converges with tolerance $\varepsilon$;
\STATE Output the state $\rho^{out} = \tr_R U(\bm\theta)\proj{00}_{AR}U(\bm\theta)^\dagger$;
\end{algorithmic}
\end{algorithm}


\section{Learning for small loss Bounds }
\label{sec: small_loss}

Recall that $L_T^* = \min_{\omega \in \cS_d} \sum_{t=1}^T \ell_t(\tr E_t \omega)$. The existing results in~\cite{Aaronson2018a} and our previous results deal with the worst-case data where the adversary is fully adaptive to the learner's strategy, which means that the $L_T^*$ can be as large as $T$ if $b_t$ is chosen arbitrarily. In reality, however, the feedback $b_t$ usually comes from the empirical mean of outcomes from single or multiple measurements we implement at each round. In particular, if we are allowed to have a relatively large amount of state copies to measure at each round, then the difference between $b_t$ and $\trace(E_t\rho)$ will be small, which means $L_T^*$ is small.  So the motivation is to design a more adaptive algorithm that can achieve small regret in such an easier case.

In this section, we obtain a "small loss" bound for $L_2$ norm, which means that the regret bound depends on $L_T^*$ instead of on $T$, so that the regret is small when $L_T^*$ is small.

\subsection{Small loss bound algorithm and results}

\begin{algorithm}[H] 
\caption{Doubling Trick for MEG and RFTL-vonNeumann}
\label{algo: small loss}
\begin{algorithmic}[1] \label{alg:main}
\STATE Input: number of rounds $T$, $\eta<\frac{1}{2}$, dimension $d$.
\STATE Set $T_1 = 1, t=1$.
\FOR[$\beta$ indexes a block]{ $\beta = 1,2,\ldots$} 
    \STATE Set $\eta_\beta = \min\{ \sqrt{\frac{n\ln 2}{2^\beta + 1}}, \frac{1}{2}\}$ \label{line: update learnint rate}
    \WHILE{ $t \leq T$ }
        \IF{$t \neq T_\beta$}
            \STATE \textbf{OPTION I:} predict \quad $\omega_t \leftarrow \textsc{MEGUPDATE}\left(\omega_{t-1},\nabla_{t-1},\eta_\beta\right)$
            \STATE \textbf{OPTION II:} predict \quad $\omega_t \leftarrow \textsc{RFTLUPDATE}\left(\{ \nabla_s\}_{s \in [T_\beta,t-1]},\eta_\beta\right)$
        \ELSE
            \STATE  predict $\omega_t  \leftarrow 2^{-n} \mathbbm{I}$
        \ENDIF
        \STATE Compute $\nabla_t = 2\left(\trace(E_{t}\omega_{t})-b_{t}\right)E_t$
        \STATE Update the cumulative loss inside current block $L_\beta = \sum_{s = T_\beta}^t\ell_s(\tr E_{s}\omega_s)$
        \IF{ $L_\beta \geq 2^\beta$} \label{line: test}
            \STATE $T_{\beta+1} \leftarrow t+1$, $t \leftarrow t+1$ \COMMENT{$T_\beta$ represent the start time of block $\beta$}
            \STATE \textbf{break}
        \ENDIF
        \STATE $t\leftarrow t+1$
    \ENDWHILE
\ENDFOR
\end{algorithmic}
\end{algorithm}

The authors of \cite{Tsuda2005} show that Matrix Exponentiated Gradient (MEG) updates can achieve the optimal small loss bound  $\operatorname{O}\left(\sqrt{nL_T^*} + n \right)$ with a pre-selected learning rate $\eta = \operatorname{O}(\sqrt{\frac{n}{L_T^*}})$. Unfortunately, in practice, $L_T^*$ is usually unknown and can even be chosen adaptively by the adversary during the learning process. 

Here we show a parameter-free small loss algorithm for both MEG and RFTL-vonNeumann by using a doubling trick to adjust the learning rate. The main idea is to use the cumulative loss $\sum_{s=1}^t \ell_s(\omega_s)$ as a proxy for the target hindsight loss $L_T^*$. But directly choosing $\eta_t = \frac{n}{\sum_{s=1}^t \ell_s(\tr E_{s}\omega_s)}$ will cause unstable updates of $\omega_t$. Therefore, we divide the time horizon into several blocks. Each time the cumulative loss within block $\beta$ exceeds $2^\beta$, the algorithm restarts into a new block as shown in line~\ref{line: test} and updates the learning rate as shown in line~\ref{line: update learnint rate}, so that $\eta_t$ always stays between $\left[\sqrt{\frac{n\ln2}{\sum_{s=1}^t \ell_s(\tr E_{s}\omega_s)+1}},\sqrt{\frac{2n\ln2}{\sum_{s=1}^t \ell_s(\tr E_{s}\omega_s)+1}}\right]$, except the first several blocks.

\begin{algorithm}[H] 
\caption{MEGUPDATE($\omega_{t-1},\nabla_{t-1},\eta_t$)}
\begin{algorithmic}[1] \label{alg:update}
    \STATE $G_{t} \leftarrow \log \omega_{t-1} - \eta_t\nabla_{t-1} $
    \STATE $\omega_t \leftarrow \frac{\exp(G_{t})}{\trace (\exp(G_t))}$
    \STATE \textbf{Return} $\omega_t$
\end{algorithmic}
\end{algorithm}

\begin{algorithm}[H] 
\caption{RFTLUPDATE$\left(\{\nabla_s\}_{s \in [T_\beta,t-1]},\eta_t\right)$}
\begin{algorithmic}[1] \label{alg:update}
    \STATE $\omega_t \leftarrow \argmin_{\omega \in \cS_d} \left\{ \eta_t \sum_{s=T_\beta}^{t-1} \left(\nabla_s \omega\right) + \trace(\omega \log \omega) \right\}$
    \STATE \textbf{Return} $\omega_t$
\end{algorithmic}
\end{algorithm}

\begin{theorem}
\label{thm: small loss}
Let $E_1,E_2,\ldots$ be a sequence of adaptively chosen two-outcome measurements on an $n$-qubit state presented to learner. The sequence of losses is defined as $\ell_t(z) = \left(z-b_t\right)^2$, where $b_t$ can also be chosen adaptively and will be presented be learner. Then algorithm~\ref{algo: small loss} with either update method guarantees a near-optimal regret 
\begin{align*}
    \Reg_T \le \operatorname{O}\left(\sqrt{n L_T^*\log T}+n\log T \right)
\end{align*}.
\end{theorem}

\begin{corollary}
\label{coro: small loss}
When $L_T^* = 0$, which is the realizable case, the regret becomes $\operatorname{O}(n\log T)$. This almost matches the lower bound $\operatorname{\Omega}(n)$ stated in \cite{Aaronson2018a}.
\end{corollary}

\subsection{Regret Analysis}
The first key step is to get the following lemma, which is inspired by \cite{Tsuda2005}, \cite{Youssry2019}.

\begin{lemma}
\label{lem: small loss}
Given the loss function $\ell_t(z) = \left( z - b_t\right)^2$ with measurement operator $-\mathbbm{I} \leq E_t \leq \mathbbm{I}$. If $\omega_t$ is the result of $\textsc{MEGUPDATE}$ or $\textsc{RFTLUPDATE}$ with learning rate $\eta \in (0,\frac{1}{2})$, then for any state $\rho$, 
\begin{align*}
    \eta \ell_{t-1}(\tr E_{t-1} \omega_{t-1}) - \frac{\eta}{1-2\eta} \ell_{t-1}(\tr E_{t-1} \rho)\leq B(\rho\|\omega_{t-1}) - B(\rho\|\omega_{t})
\end{align*}
where $B(X\|Y) = \trace\left(X\ln X - X \ln Y \right)$, which is the Bregman divergence of von Neumann entropy.

\end{lemma}

\begin{remark}
This lemma indicates that MEG and RFTL-vonNeumann is inherently similar. Note that MEG is modified from online mirror descent (OMD) with von Neumann Entropy and RFTL is well known to have a strong connection to OMD (see, e.g., \cite{Hazan2016,Shalev-Shwartz2007,Abernethy2009}). So this lemma matches our intuition behind this connection.
\end{remark}

Then for any block $\beta$ except the last one, by rearranging the equation in Lemma~\ref{lem: small loss} and summing up all the losses within block $\beta$, we have
\begin{align*}
    \sum_{t = T_{\beta}}^{T_{\beta+1}-1} \ell_{t}(\tr E_{t}\omega_{t}) - \ell_{t}(\tr E_{t} \rho) 
    \leq \frac{1}{\eta_\beta}\max_\rho \left\{B(\rho\|2^{-n} \mathbbm{I}) \right\} + 2\eta_\beta \sum_{t=T_\beta}^{T_{\beta+1}-1} \ell_t(\tr E_{t} \omega_t)
\end{align*}

Our learning rate adjusting strategy guarantees that this is always upper bounded by $\order\left( \sqrt{\sum_{t=T_\beta}^{T_{\beta+1}-1} \ell_t(\tr E_{t}\omega_t)}\right)$. So by summing up the losses over all blocks gives us the result. We postpone a complete version of proof to Appendix~\ref{app: small loss}.

\section{Numerical experiments}

In this section, we evaluate the RFTL-Tallis-2 and Doubling Trick algorithms with a series of numerical experiments.  For the first algorithm, we compare the performances of RFTL-Tallis-2 under the full rank measurement and 1-rank measurement, to the performances of RFTL-vonNeumann under those two measurements; For the second algorithm, we compare the performances of RFTL-vonNeumann and MEG equipped with doubling trick (TD) learning rate adjustment strategy to those unequipped.  

\paragraph{Feedbacks}  Suppose there is an underlying unknown quantum state $\rho$, pure or mixed, to be learned. We first consider the noiseless realizable setting for both algorithms. In each round, the learner will observe $b_t = \tr(E_t\rho)$, which is the exact probability of getting "Yes" measurement outcome. In addition, we also consider the noisy non-realizable settings for the small loss experiment. To be specific, we assume that a certain state copy is available for testing at each round, and each test will give us a Bernoulli feedback with "Yes" probability $\tr(E_t\rho)$, with some added small noise. Here we choose the parameters $100$ state copies and $0.05*N(0,0.1)$ Gaussian noise.

\paragraph{Adversary's data generation strategy} For the adaptive adverserial setting, generating a worst-case data sequence to fool an algorithm is challenging. Here we first use the same method in \cite{Yang2020} to generate $E_t = \argmax_E \left( \tr(E(\omega_t - \rho)) \right)^2$ at each round. Moreover, we observe that rank-1 measurement can only maximize the loss in one direction. So if we randomly generate a mixed state $\rho$, then the losses in rank-1 measurement settings are likely to be much smaller than ones in full rank measurement settings, which is not easy for numerical comparison. Therefore we set $\rho$ to be pure so that the spectral norm of $\omega_t - \rho$ becomes larger than that in the mixed state case, and therefore the losses are of the same order of magnitude as those in full rank measurement settings.

\paragraph{Other parameters} We choose the $L_2$ norm loss. We take the number of qubits $n = 4$, so the dimension of the density matrices is $16 \times 16$. For each experiment, we run for 100 trials with randomly generated target states $\rho$ and report the average curves.

\subsection{Experiment Results}

The first experimental results are illustrated in Figure~\ref{fig: experiment} (a). It is easy to see that when the rank of measurement is $1$, Tallis-2 Entropy regularizer takes the advantage over Von Neumann regularizer. On the other hand, when the measurement is full rank, these two algorithms perform almost the same. Notice that theoretically, Von Neumann regularizer should take advantage over Tallis-2 Entropy regularizer in the full rank when $n$ is large. But due to the device limitation, we leave the experiment on high dimensional states for future study.


The second experimental results are illustrated in Figure~\ref{fig: experiment} (b)(c). The results clearly show that the doubling trick makes both RFTL-vonNeumann and MEG more adaptive to the small hindsight loss. Besides, we also observe that the MEG and RFTL-vonNeumann have very close performances, which matches our intuition in Lemma~\ref{lem: small loss}.

\begin{figure}[H]
	\centering
	\begin{adjustwidth}{0.1cm}{0.1cm}
		\begin{tikzpicture}
 		\node at (-5.0,3.0) {\includegraphics[width = 4.55cm]{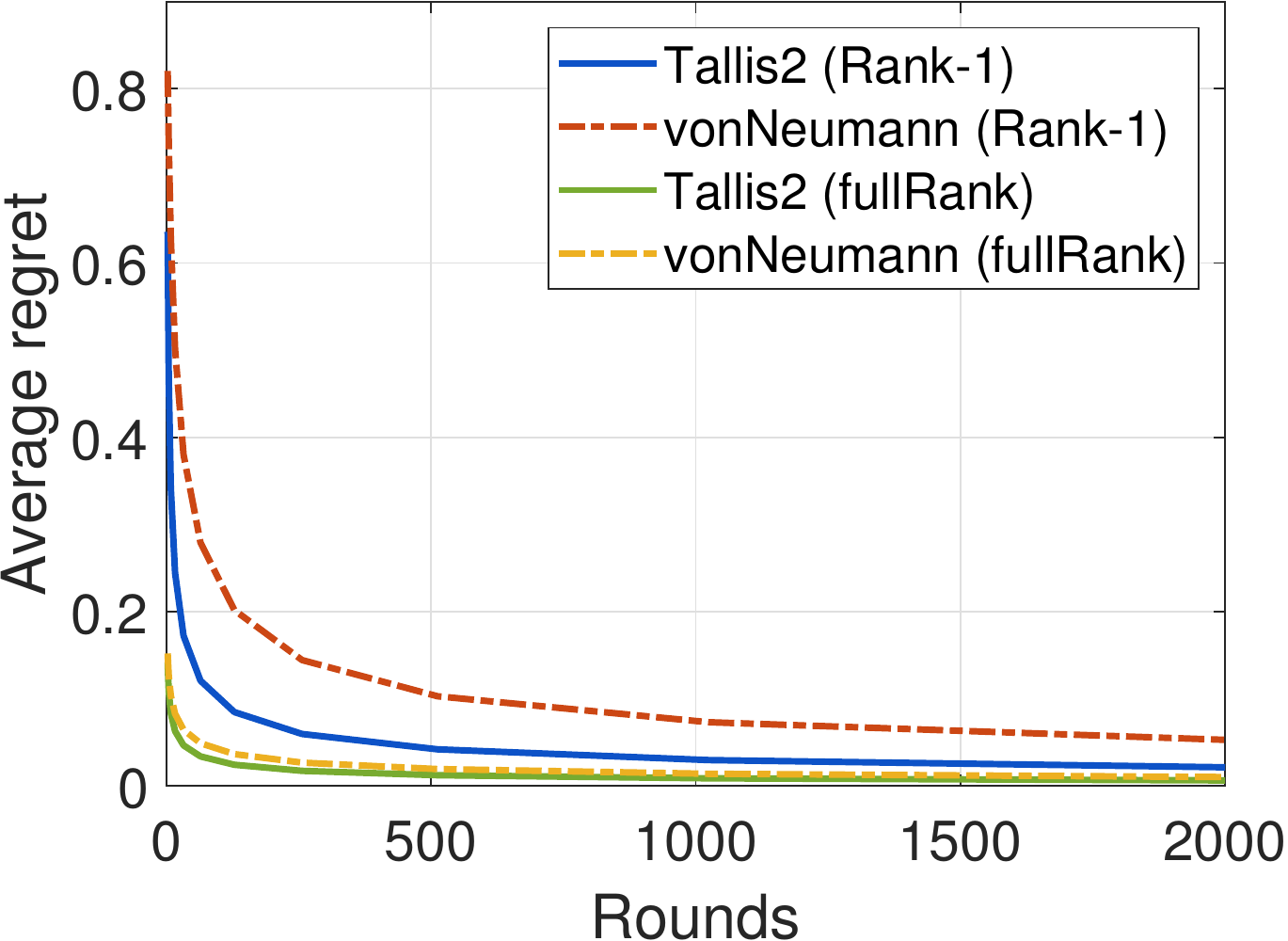}};
		\node at (-5,0.6) {\small (a)  vonNeumann vs. Tallis-2};
		
		\node at (-0.1,3.0) {\includegraphics[width = 4.55cm]{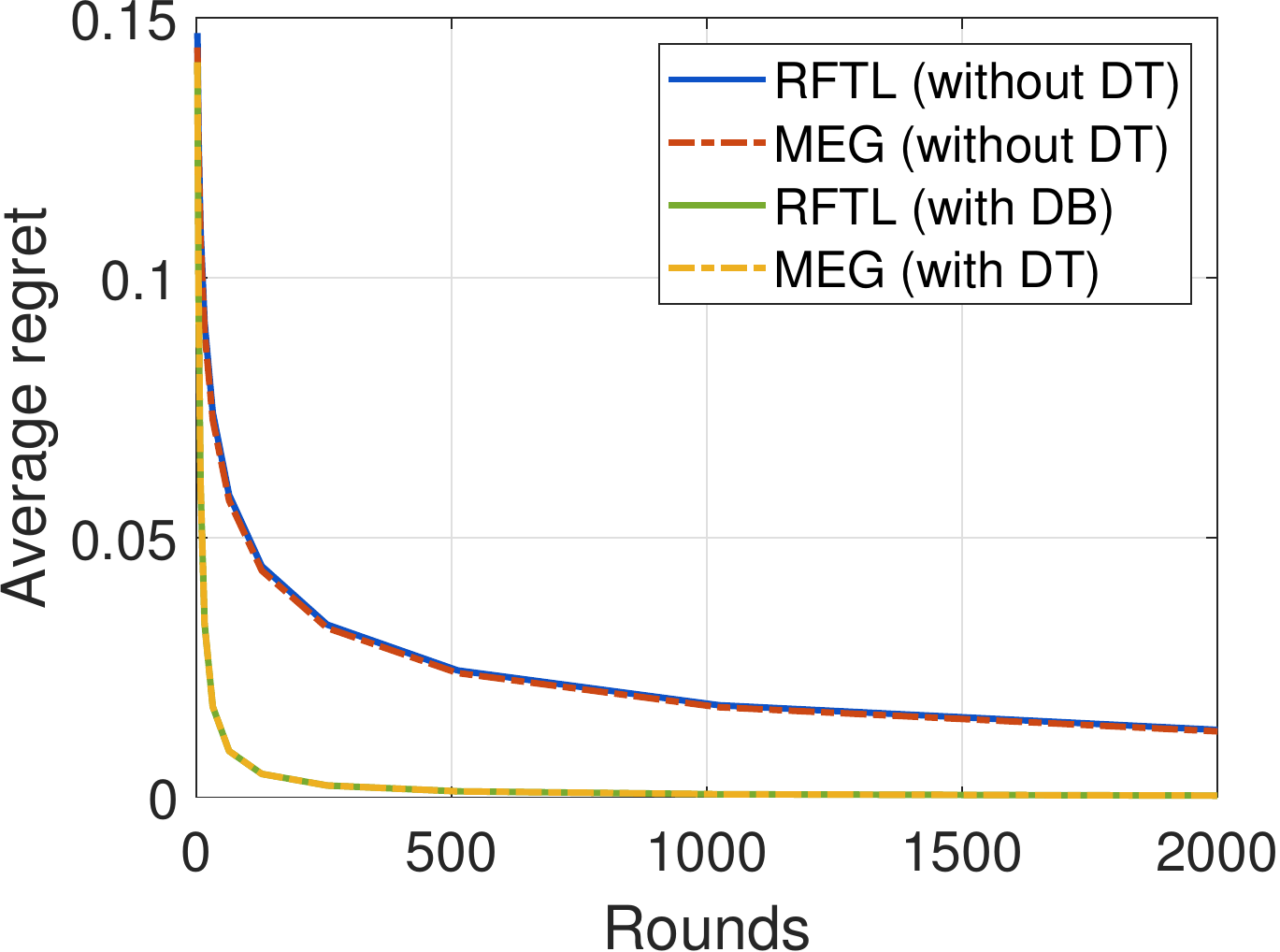}};
		\node at (0.05,0.6) {\small (b) small loss algorithm with $L^* =0$};
		
		\node at (4.8,3.0) {\includegraphics[width = 4.55cm]{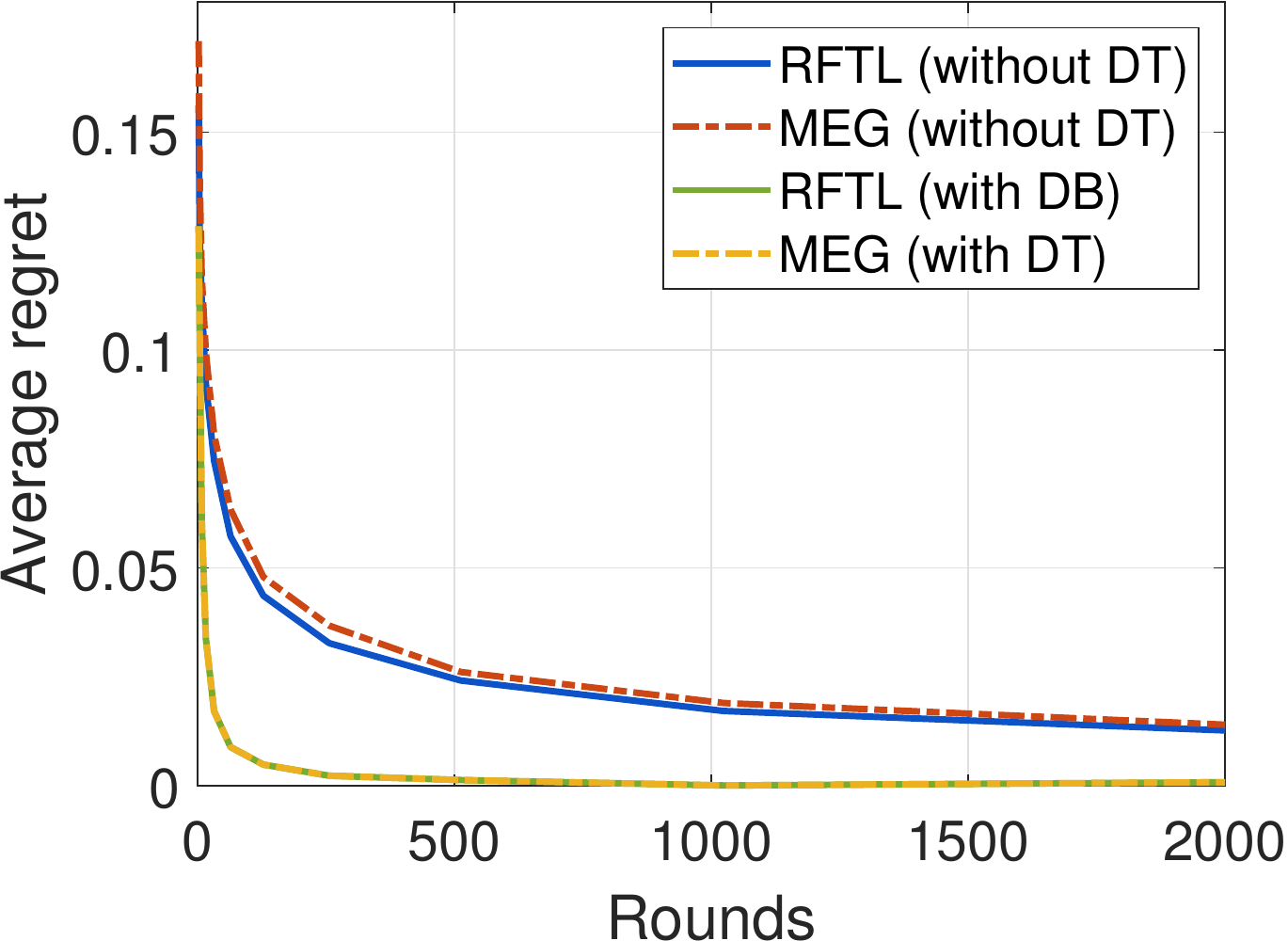}};
		\node at (5.3,0.6) {\small (c) small loss algorithm with $L^*$ is small};
		
		\end{tikzpicture}
	\end{adjustwidth}
\caption{The experimental results for  online learning of $4$-qubit states with average regret as the metric.}
\label{fig: experiment}
\end{figure}
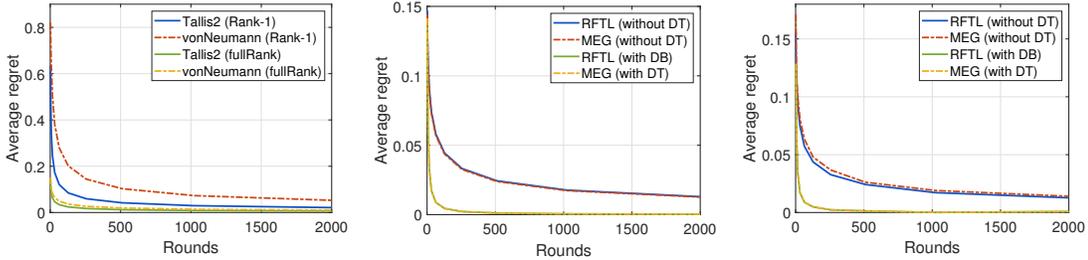



\section{Summary}
In summary, we have developed more practical and adaptive algorithms for online state learning. We have proposed an RFTL algorithm with Tallis-2 entropy that can guarantee to predict the quantum states with regret bounded by $\order(\sqrt{MT})$ for the first $T$ measurements with maximum rank $M$. This regret bound of our algorithm has advantages over previous methods in the setting with low-rank measurements since it only depends on the rank of the measurements rather than the number of qubits. In particular, we have shown a variational hybrid quantum-classical algorithm for state prediction in the RFTL algorithm, which could be run on near-term quantum computers. We have also proposed a parameter-free algorithm that can achieve a regret depending on the loss of best states in hindsight.

For future work, it will be essential to develop tight generalizations of our results to measurements with multiple outcomes. Another interesting direction is to find lower and upper bounds in different practical and adaptive regimes, which could help us better understand the fundamental limits of online state learning. 

\textbf{Acknowledgements}: This work was done when YC was a research intern at Baidu Research.

\newpage

\bibliographystyle{alpha}
\bibliography{smallbib}

\appendix 

\section{Analysis for Algorithm~\ref{alg:rftl} and other related results}
\label{app: tallis2}

\subsection{SDP for the prediction subroutine (Proposition~\ref{prop: sdp})}
We have
\begin{align}
&{ \min_{\omega\in \cS_d} }\left\{\eta \sum_{s=1}^{t-1} \tr\left(\nabla_{s} \omega\right)-S_2(\omega)\right\} \\
= &{ \min_{\omega\in \cS_d} }\left\{\eta \sum_{s=1}^{t-1} \tr\left(\nabla_{s} \omega\right) + \tr \rho^2 - 1\right\} \\
 = &\min_{\omega\in \cS_d}\left\{\eta \sum_{s=1}^{t-1} \tr\left(\nabla_{s} \omega\right) + \tr Q - 1: Q\ge \rho I^{-1} \rho \right\}\\
= & {\min }\left\{\eta \sum_{s=1}^{t-1} \tr\left(\nabla_{s} \omega\right)+ \tr Q -1: 
\left[ {\begin{array}{*{20}{c}}
Q&\omega \\
\omega &I
\end{array}} \right] \ge 0, \omega\ge0, \tr\omega=1\right\},
\end{align}
where the last equality uses the Schur complement characterization of the block positive semidefinite operator.

\subsection{Proof of Theorem~\ref{thm:regret bound} and other related results}
\begin{proof} \textbf{of Theorem~\ref{thm:regret bound}}
Note that the loss function is convex, then for any state $\varphi$,
\begin{align}
\ell_t(\tr E_t\omega_t) - \ell_t(\tr E_t\varphi) \le \ell_t'(\tr E_t\omega_t)[\tr E_t\omega_t -\tr E_t\varphi] = \tr\nabla_t(\omega_t-\varphi).
\end{align}
Taking the sum over $t$, we have
\begin{align}
\sum_{t=1}^T [\ell_t(\tr E_t\omega_t) - \ell_t(\tr E_t\varphi)] &\le \sum_{t=1}^T [\tr \nabla_t\omega_t - \tr \nabla_t\varphi]\\
& \le \sum_{t=1}^T [\tr \nabla_t(\omega_t -\omega_{t+1}) + \frac{D_S^2}{\eta}], \label{eq:regret upper bound DR}
\end{align}
where $D_S := \max_{\varphi_1,\varphi_2\in\cS_d}[S_2(\varphi_1) - S_2(\varphi_2)]$. Moreover, the last inequality is due to Lemma 5.3 in \cite{Hazan2016}. Using the definition of Tsallis-$2$ entropy, i.e., $S_2(\varphi) = 1-\tr\varphi^2$,  it is not difficult to find that $D_S = 1-\frac{1}{d}$ and the maximizer is $\varphi_1=I/d$, $\varphi_2=\proj 0$. Therefore, it holds that $D_s\le 1$.

The next step is to bound the quantity  $\sum_{t=1}^T \tr \nabla_t(\omega_t -\omega_{t+1})$. Let us denote 
\begin{align}\label{eq:phi t}
\Phi_t(X) = \eta \sum_{s=1}^t \tr\nabla_sX + \tr X^2 -1
\end{align}
and assume that the optimal solution or minimizer of $\min_{X\in\cS_n}\Phi_t(X)$ is $\omega_{t+1}$. We further introduce 
\begin{align}
\nabla{\Phi_t}(X):=\eta\sum_{s=1}^t\nabla_s +2X.
\end{align}

Recall the definition of the Bregman divergence, we have that 
\begin{align}\label{eq:bregman norm}
B_{\Phi_t}(X\|Y)&= \Phi_t(X) - \Phi_t(Y) - \tr\nabla\Phi_t(Y)(X-Y) \\
& = \tr (X-Y)^2\\
& = \|X-Y\|_2^2.
\end{align}
Indeed, for any $t$, we have   $B_{\Phi_t}(X\|Y) = \|X-Y\|_2^2$.
 
Furthermore, taking $X=\omega_t$ and $Y=\omega_{t+1}$, we have that
\begin{align}
\Phi_t(\omega_t) &= \Phi_t(\omega_{t+1}) + \tr (\omega_t-\omega_{t+1})\nabla\Phi_t(\omega_{t+1}) + B_{\Phi_t}(\omega_t\|\omega_{t+1})\\
& \ge \Phi_t(\omega_{t+1}) + B_{\Phi_t}(\omega_t\|\omega_{t+1}). \label{eq:B upper}
\end{align}
The above inequality follows due to Lemma~\ref{lem: KKT}, which proves that $\tr (\omega_t-\omega_{t+1})\nabla\Phi_t(\omega_{t+1})\ge 0$.

Therefore, we have
\begin{align}
\|\omega_t-\omega_{t+1}\|_2^2 & = B_{\Phi_t}(\omega_t\|\omega_{t+1})\\ & \le \Phi_t(\omega_t) - \Phi_t(\omega_{t+1})\\
& = \Phi_{t-1}(\omega_t) - \Phi_{t-1}(\omega_{t+1}) + \eta \tr \nabla_t(\omega_t - \omega_{t+1})\\ 
&\le \eta\tr \nabla_t(\omega_t - \omega_{t+1}). \label{eq:B Dt w}
\end{align}
The first inequailty follows due to Eq.~\eqref{eq:B upper}. The last inequality follows since $\omega_t$ is the minimizer.

Applying the generalized Cauchy-Schwartz inequality~\cite{Bhatia2013}, we have
\begin{align}
\tr \nabla_t(\omega_t-\omega_{t+1}) &\le \|\nabla_t\|_2\|\omega_t-\omega_{t+1}\|_2\\
& \le \|\nabla_t\|_2\cdot \sqrt{\tr \eta\nabla_t(\omega_t - \omega_{t+1})},
\end{align}
where the last inequality utilizes the inequality in Eq.~\eqref{eq:B Dt w}. Immediately, noting that $\lambda= \max_t \|E_t\|_2$, we could obtain
\begin{align}
\tr \nabla_t(\omega_t - \omega_{t+1}) \le \eta \|\nabla_t\|^2_2
\le \eta \lambda^2L^2,\label{eq: Delta t upper bound}
\end{align}
where the last inequality follows since $\|\nabla_t\|_2 = \|\ell'_t(\tr E_t\omega_t)E_t\|_2\le L\|E_t\|_2\le L\lambda$.

Combining Eq.~\eqref{eq:regret upper bound DR} and Eq.~\eqref{eq: Delta t upper bound}, we have
\begin{align}
\sum_{t=1}^T [\ell_t(\tr E_t\omega_t) - \ell_t(\tr E_t\varphi)] & \le  \sum_{t=1}^T [\tr \nabla_t(\omega_t -\omega_{t+1}) + \frac{D_R^2}{\eta}] \\
& \le \eta T \lambda^2L^2 + \frac{1}{\eta}.
\end{align} 
By setting $\eta = \frac{1}{L\lambda\sqrt{T}}$, we have
\begin{align}
\sum_{t=1}^T [\ell_t(\tr E_t\omega_t) - \ell_t(\tr E_t\varphi)]   
& \le 2L\lambda\sqrt{T}.
\end{align}

\end{proof}

\begin{proof} \textbf{of Corollary~\ref{coro: tallis2}}
    The inequality in Eq.~\eqref{eq:rank bound} follows since $\|E\|_2\le \sqrt M$ for a rank-$M$ measurement operator.
    The inequality in Eq.~\eqref{eq:dimension bound} follows since the rank of an arbitrary measurement operator is bounded by the dimension $d$.
\end{proof}

\begin{theorem}[Explaination of Remark~\ref{remark: relation with OGD}]
By running the OGD algorithm in \cite{Yang2020} with $\eta = \frac{1}{L\sqrt{MT}}$, we can also get
\begin{align*}
     \Reg_T \leq \operatorname{O}(L\sqrt{MT})
\end{align*}
\end{theorem}
\begin{proof}
    The proof follows the same manner in Appendix B \cite{Yang2020} but again we use the property of Frobenius norm. 
    
    According to their proof, we will get
    \begin{align*}
        \tr(\nabla_t (\omega_t - \varphi)) 
        &\leq  \frac{\|\omega_t - \varphi\|_F^2 - \|\omega_{t+1} - \varphi\|_F^2}{2\eta} + \frac{1}{2}\eta \|\nabla_t\|_F^2\\
        &\leq  \frac{\|\omega_t - \varphi\|_F^2 - \|\omega_{t+1} - \varphi\|_F^2}{2\eta} + \frac{1}{2}\eta L^2 M
    \end{align*}
    So by choosing $\eta = \frac{1}{L\sqrt{dT}}$, we get
    \begin{align*}
        \Reg_T \leq \frac{1}{2\eta} \max_\varphi \|\omega_1 -\varphi\|_F^2  + \frac{1}{2}\eta L^2 M \leq \operatorname{O}(L\sqrt{MT})
    \end{align*}
\end{proof}

\subsection{Auxiliary Lemma}
\begin{lemma}
\label{lem: KKT}
For all $t \in \{1,2,\ldots,T \}$ and $\Phi_t = \eta\sum_{s=1}^t \tr \nabla_sX + \trace X^2 - 1$  
\begin{align}
    \tr (\omega_t-\omega_{t+1})\nabla\Phi_t(\omega_{t+1})\ge 0
\end{align}
\end{lemma}
\begin{proof}
This proof again along the lines of Claim 14 in \cite{Aaronson2018a}. The only difference is that, instead of having 
\begin{align*}
    \frac{\Phi_t(\overline{X}) - \Phi_t(\omega_{t+1})}{a} \le \nabla \Phi_t(\omega_{t+1}) \cdot (\omega_t - \omega_{t+1}) + \frac{\tr((\omega_t - \omega_{t+1})^2)}{\lambda_{\min}(\omega_{t+1)}}
\end{align*}
we have
\begin{align*}
    \frac{\Phi_t(X) - \Phi_t(\omega_{t+1})}{a}
    & \le \nabla \Phi_t(\omega_{t+1}) \cdot (\omega_t - \omega_{t+1}) + \frac{B_{\Phi_t}(\overline{X}\| \omega_{t+1})}{a}\\
    & = \nabla \Phi_t(\omega_{t+1}) \cdot (\omega_t - \omega_{t+1}) + \tr(\Delta^2)\\
    & = \nabla \Phi_t(\omega_{t+1}) \cdot (\omega_t - \omega_{t+1}) + a \tr((\omega_t - \omega_{t+1})^2)
\end{align*}
where $\Delta$ is defined the same as in that Claim 14. Again we can find $a$ small enough to show the contrast as in the previous proof.
\end{proof}

\section{Analysis for Algorithm~\ref{algo: small loss} and other related results}
\label{app: small loss}

\subsection{Proof of Lemma~\ref{lem: small loss}} 
Before the proof, we want to remind that here $B(X\|Y) = \trace\left(X\ln X - X \ln Y \right)$, which is the Bregman divergence of von Neumann entropy. It should be actually represented as $B_{\Phi_t}(X\|Y)$ but we simplified the notation. This is different from the previous proof, where $B_{\Phi_t}(X\|Y)$ is the Bregman divergence of Tallis-2 entropy.

The proof is inspired by Lemma 2 \cite{Youssry2019} and Lemma 3.1 \cite{Tsuda2005}. In particular, their proofs only work for MEG algorithm and we generalized their proofs to both MEG and RFTL-vonNeumann. 

For convenience, suppose $t \in [T_\beta,T_{\beta+1}-1]$ for some $\beta$ so that $\eta_t = \eta_\beta$. Also let  $\delta_t = -2\eta_\beta(\trace(E_t\omega_t - b_t))$ in all the following proofs.
\begin{lemma}
\label{lem: variance1}
If $\omega_t = \textsc{MEGUPDATE}\left(\omega_{t-1},\nabla_{t-1},\eta_t \right)$, then for any state $\varphi$,
\begin{align*}
    B(\varphi\|\omega_t) - B(\varphi\|\omega_{t+1}) \geq \delta_t \trace\left(\varphi E_t \right) - \log \left( \trace\left[ \exp(\delta_t E_t)\omega_{t} \right] \right)
\end{align*}
\end{lemma}

\begin{proof}
\begin{align*}
    B(\varphi\|\omega_t) - B(\varphi\|\omega_{t+1})
    & = \trace(\varphi \log(\varphi) - \varphi\log(\omega_t) - \trace(\varphi \log(\varphi) - \varphi\log(\omega_{t+1}))\\
    & = -\trace(\varphi \log(\omega_t)) + \trace(\varphi \log(\omega_{t+1}))\\
    & = -\trace(\varphi \log(\omega_t)) + \trace(\varphi \log(\omega_{t})) + \trace(-\varphi \eta_\beta \nabla_t) - \trace(\varphi \log(\trace(\exp(\log(\omega_t)-\eta_\beta \nabla_t)))\\
    & = \delta_t \trace(\varphi E_t) - \log(\trace(\exp(log(\omega_t)+\delta_tE_t)))\\
    & \geq \delta_t \trace\left(\varphi E_t \right) - \log \left( \trace\left(\exp(\delta_t E_t)\omega_{t} \right) \right) \quad \text{( By Golden-Thompson Inequality) }
\end{align*}
The third equation comes from the definition of $\omega_{t+1}$
\end{proof}

\begin{lemma}
\label{lem: variance2}
If $\omega_t =\textsc{RFTLUPDATE}\left(\{\nabla_s\}_{s \in [T_\beta,t-1]},\eta_t\right)$, then for any state $\varphi$,
\begin{align*}
    B(\varphi\|\omega_t) - B(\varphi\|\omega_{t+1}) = \delta_t \trace\left(\varphi E_t \right) - \log \left( \trace\left[ \exp(\delta_t E_t)\omega_{t} \right] \right)
\end{align*}
\end{lemma}

\begin{proof}
It is well known that this optimization can be written as a closed form as 

\begin{align*}
    \omega_{t+1} = \frac{\exp(-\eta_\beta\sum_{s=T_\beta}^t\nabla_s)}{\trace(\exp(-\eta_\beta\sum_{s=T_\beta}^t\nabla_s))}
\end{align*}
Now we have
\begin{align*}
    \log(\omega_{t+1}) -  \log(\omega_{t})
    & = \log \left( \frac{\exp(-\eta_\beta\sum_{s=T_\beta}^t\nabla_s)}{\trace(\exp(-\eta_\beta\sum_{s=T_\beta}^t\nabla_s))}\right)
    - \log \left( \frac{\exp(-\eta_\beta\sum_{s=T_\beta}^{t-1}\nabla_s)}{\trace(\exp(-\eta_\beta\sum_{s=T_\beta}^{t-1}\nabla_{s}))}\right) \\
    & =  -\eta_\beta \nabla_t - \log \left( \frac{\trace(\exp(-\eta_\beta\sum_{s=T_\beta}^t\nabla_s))}{\trace(\exp(-\eta_\beta\sum_{s=T_\beta}^{t-1}\nabla_{s}))}\right)\\
    & = -\eta_\beta \nabla_t - \log \left(\frac{\trace\left(\exp(-\eta_\beta \nabla_t)\omega_{t}*\trace(\exp(-\eta_\beta\sum_{s=T_\beta}^{t-1}\nabla_{s}))\right)}{\trace(\exp(-\eta_\beta\sum_{s=T_\beta}^{t-1}\nabla_{s}))} \right)\\
    & = - \eta_\beta \nabla_t - \log \left( \trace\left( \exp(-\eta_\beta \nabla_t)\omega_{t} \right)\right) 
\end{align*}
Therefore,
\begin{align*}
    B(\varphi\|\omega_t) - B(\varphi\|\omega_{t+1})
    & =-\trace(\varphi \log(\omega_t)) + \trace(\varphi \log(\omega_{t+1}))\\
    & = \trace\left(\varphi \left( -\eta_\beta \nabla_t - \log \left( \trace\left( \exp(-\eta_\beta \nabla_t)\omega_{t} \right)\right) \right)\right)\\
    & = \trace\left(-\varphi \eta_\beta \nabla_t \right) - \log \left( \trace\left( \exp(-\eta_\beta \nabla_t)\omega_{t} \right) \right)\\
    & = \delta_t \trace\left(\varphi E_t \right) - \log \left( \trace\left( \exp(\delta_t E_t)\omega_{t} \right) \right)
\end{align*}
\end{proof}

\begin{lemma}[Lemma 1~\cite{Youssry2019}]
\label{lem: nomalization_factor}
\begin{align*}
    \delta_t \trace\left(\varphi E_t \right) - \log \left( \trace\left[ \exp(\delta_t E_t)\omega_{t} \right] \right)
    \leq \frac{\delta_t^2}{2}+\delta_t \trace(\omega_t E_t)
\end{align*}
\end{lemma}
\begin{proof}
    The proof is exactly the same as in Lemma 1~\cite{Youssry2019}. Except that we already get (B2) inequality in their proof.
\end{proof}

Now we are ready to prove Lemma~\ref{lem: small loss},
\begin{align*}
    \eta_\beta \ell_{t-1}(\tr E_{t-1}\omega_{t-1}) - \frac{\eta_\beta}{1-2\eta_\beta} \ell_{t-1}(\tr E_{t-1}\varphi) 
    & \leq \delta_t \trace\left(\varphi E_t \right) - \frac{\delta_t^2}{2}-\delta_t \trace(\omega_t E_t) \\
    & \leq \delta_t \trace\left(\varphi E_t \right) - \log \left( \trace\left( \exp(\delta_t E_t)\omega_{t} \right) \right)\\
    & \leq  B(\varphi\|\omega_t) - B(\varphi\|\omega_{t+1})
\end{align*}
The first inequality comes from (B11)-(B21)~\cite{Youssry2019}. The second inequality comes from Lemma~\ref{lem: nomalization_factor}. The last inequality comes from Lemma~\ref{lem: variance1} and Lemma~\ref{lem: variance2}.

\subsection{Proof of Theorem~\ref{thm: small loss} }

For convenience, we denote $L_\beta:=\sum_{t=T_\beta}^{\min\{T_{\beta+1}-1,T\}} \ell_t(\tr E_t\omega_t)$, $L^\varphi_\beta:=\sum_{t=T_\beta}^{\min\{T_{\beta+1}-1,T\}} \ell_t(\tr E_t\varphi)$, $L_T:=\sum_{t=1}^{T} \ell_t(\tr E_t\omega_t)$ and $L^\varphi_T:=\sum_{t=1}^{T} \ell_t(\tr E_t\varphi)$.

For any $t \in [T_\beta+1,T_{\beta+1}-1]$, by rearranging the equation in Lemma~\ref{lem: small loss}, we have 
\begin{align*}
    \ell_{t-1}(\tr E_{t-1}\omega_{t-1}) - \ell_{t-1}(\tr E_{t-1}\varphi)\leq \frac{1 - 2\eta_\beta}{\eta_\beta}\left[B(\varphi\|\omega_{t-1}) - B(\varphi\|\omega_{t})\right] + 2\eta_\beta \ell_{t-1}(\tr E_{t-1}\omega_{t-1})
\end{align*}
When $t = T_{\beta+1}$, because we always reset the $\omega$ to $2^{-n} \mathbbm{I}$ at the beginning of new block, so we cannot use this lemma. But because $B(X\|Y) \geq 0$ for all feasible $X,Y$, so we can simply bound that by 
\begin{align*}
    \ell_{t-1}(\tr E_{t-1}\omega_{t-1}) - \ell_{t-1}(\tr E_{t-1}\varphi)\leq \frac{1 - 2\eta_\beta}{\eta_\beta}B(\varphi\|\omega_{t-1})  + 2\eta_\beta \ell_{t-1}(\tr E_{t-1}\omega_{t-1})
\end{align*}
By summing up all the losses within block $\beta$, we have
\begin{align*}
    L_\beta - L_\beta^\varphi
    &=\sum_{t = T_{\beta}+1}^{T_{\beta+1}} \ell_{t-1}(\tr E_{t-1}\omega_{t-1}) - \ell_{t-1}(\tr E_{t-1}\varphi) \\
    & \leq \sum_{t = T_{\beta}+1}^{T_{\beta+1}}\left( \frac{1 - 2\eta_\beta}{\eta_\beta}\left(B(\varphi\|\omega_{t-1}) - B(\varphi\|\omega_{t})\right) + 2\eta_\beta \ell_{t-1}(\tr E_{t-1}\omega_{t-1}) \right) + \frac{1 - 2\eta_\beta}{\eta_\beta}\left[B(\varphi\|\omega_{T_{\beta+1}})\right]\\
    & = \frac{1 - 2\eta_\beta}{\eta_\beta}B(\varphi\|\omega_{T_\beta}) +2\eta_\beta \sum_{t = T_{\beta}}^{T_{\beta+1}-1}\ell_{t}(\tr E_t\omega_{t})\\
    & \leq \underbrace{\frac{1}{\eta_\beta}}_\text{Term 1}\max_{\varphi_1} \left\{B(\varphi_1\|2^{-n} \mathbbm{I}) \right\} + \underbrace{2\eta_\beta L_\beta}_\text{Term 2}
\end{align*}
Let's first summing up Term 1 over all the blocks.

Suppose we will have $\beta_{max}$ blocks in total. If $\beta_{max}\geq 2$  then
\begin{align*}
    \sum_{\beta=1}^{\beta_{max}} \max \left\{ \sqrt{\frac{2^\beta + 1}{n\ln 2}},2\right\} 
    & \leq \frac{1}{\sqrt{n\ln 2}} \sqrt{\beta_{max}\sum_{\beta=1}^{\beta_{max}}(2^\beta +1)}+ 2\beta_{max}\\
    & \leq \frac{1}{\sqrt{n\ln 2}} \sqrt{\beta_{max}2^{\beta_{max}+1}+\beta_{max}^2}+ 2\beta_{max}\\
    & \leq \frac{1}{\sqrt{n\ln 2}} \sqrt{4\beta_{max}\sum_{t=T_{\beta_{max}-1}}^{T_{\beta_{max}}-1}\ell_t(\tr E_t\omega_t)+\beta_{max}^2}+ 2\beta_{max}\\
    & \leq \frac{1}{\sqrt{n\ln 2}} \sqrt{4\beta_{max}L_T+\beta_{max}^2}+ 2\beta_{max}
\end{align*}
The second inequality comes from Cauchy-Schwarz inequality and the penultimate inequality comes Line~\ref{line: test} in the algorithm.
\\\\
If $\beta_{max} = 1$, which means the algorithm never starts a new block, we have $\eta_1 = \min\{ \sqrt{\frac{n \ln 2}{2^\beta + 1}}, \frac{1}{2}\} = \frac{1}{2}$, therefore the upper bound is just $2$
\\\\
Next we sum up Term 2 over all the blocks,
\begin{align*}
    2 \sum_{\beta=1}^{\beta_{max}} \min \left\{ \sqrt{\frac{n\ln 2}{2^\beta+1}} , \frac{1}{2} \right\}L_\beta
    & \leq 2 \sum_{\beta=1}^{\beta_{max}} \sqrt{\frac{n\ln 2}{2^\beta+1}} L_\beta\\
    & \leq 2 \sum_{\beta=1}^{\beta_{max}} \sqrt{\frac{n\ln 2}{\sum_{s = T_{\beta}}^{T_{\beta+1}-2} \ell_s(\tr E_s\omega_s)+1}} L_\beta \\
    & \leq 2 \sum_{\beta=1}^{\beta_{max}} \sqrt{\frac{n\ln 2}{L_\beta}} L_\beta \\
    & \leq 2 \sqrt{n(\ln2) \beta_{max} L_T}
\end{align*}
The second inequality again comes from Line~\ref{line: test} in the algorithm. Because the algorithm does not restart a new block at time $T_{\beta+1}-2$, so $\sum_{t = T_\beta}^{T_{\beta+1}-2} \ell_t(\omega_t) < 2^\beta$. The penultimate inequality comes form the fact that $\ell_{T_{\beta+1}-1}(\omega_{T_{\beta+1}-1}) \leq 1$. And the last inequality again comes from Cauchy-Schwarz inequality.

Therefore. combine the sum of Term 1 and sum of Term 2 obtained from above.
\begin{align*}
    L_T - L_T^\varphi 
    & \leq \frac{\max_\varphi \left\{B(\varphi\|2^{-n} \mathbbm{I}) \right\}}{\sqrt{n\ln 2}} \sqrt{4\beta_{\max} L_T+\beta_{\max}^2(2T)} + 2 \sqrt{n(\ln 2) \beta_{\max} L_T} + 2\beta_{\max}\max_{\varphi_1} \left\{B(\varphi_1\|2^{-n} \mathbbm{I}) \right\} 
\end{align*}

Finally, $\max_{\varphi_1} \left\{B(\varphi_1\|2^{-n} \mathbbm{I}) \right\} = n \ln 2$ and there are at most $\log(2T)$ blocks because
\begin{align*}
    2^{\beta_{max}} = 2* 2^{\beta_{max}-1} \leq 2 L_{\beta_{max}-1} \leq 2T
\end{align*}
So,
\begin{align*}
    L_T - L_T^\varphi 
    \leq 3log(2T)(\ln2) n + 4 \sqrt{(n\ln 2) \log(2T) L_T}
\end{align*}
Solving for $\sqrt{L_T}$ leads to
\begin{align*}
    \sqrt{L_T} \leq \sqrt{3log(2T)\ln(2) n + 4(n\ln 2 \log(2T)) + L^\varphi_T} + 2 \sqrt{n\ln 2 \log(2T)}
\end{align*}
Squaring both side and doing some rearrangement and relaxation give,
\begin{align*}
    L_T - L_T^\varphi
    &\leq 3log(2T)\ln(2) n + 16(n\ln 2 \log(2T)) + 4\sqrt{n\ln 2 \log(2T)}\sqrt{3log(2T)\ln(2) n} + 4 \sqrt{n\ln 2 \log(2T) L_T^\varphi}\\
    & \leq \operatorname{O}\left(nlog(T) + \sqrt{n\log T L_T^\varphi}\right)
\end{align*}
Note that $\varphi$ can be any state, so this result certainly holds for $L_T^\varphi= L_T^*$.

%
%
%
%
%

\end{document}